\documentclass{article}
\usepackage{times}  
\usepackage{helvet} 
\usepackage{courier}
\usepackage{caption}
\usepackage[T1]{fontenc}                       % use 8-bit T1 fonts
\usepackage{hyperref}                          % hyperlinks
\usepackage{url}                               % simple URL typesetting
\usepackage{booktabs}                          % professional-quality tables
\usepackage{amsfonts}                          % blackboard math symbols
\usepackage{amssymb}                           % enable \mathbb symbols
\usepackage{amsmath}                           % allow custom math operators
\usepackage{nicefrac}                          % compact symbols for 1/2, etc.
\usepackage{microtype}                         % microtypography
\usepackage{graphicx}                          % include images
\usepackage[ruled,linesnumbered]{algorithm2e}  % pseudocode
\usepackage[inline]{enumitem}                  % inline lists
\usepackage[numbers]{natbib}                   % cite

\usepackage{bbm}
\usepackage{mathtools}
\usepackage{nccmath}
\usepackage{microtype}
\usepackage{multirow}
\usepackage{csquotes}
\usepackage{tabularx,booktabs}
\newcolumntype{Y}{>{\centering\arraybackslash}X}
\DeclareCaptionStyle{ruled}{labelfont=normalfont,labelsep=colon,strut=off}

\usepackage{hyperref}                          % hyperlinks
\usepackage{url}                               % simple URL typesetting
\usepackage{tabularx}
\usepackage{booktabs}                          % professional-quality tables
\usepackage{amsfonts}                          % blackboard math symbols
\usepackage{amsmath}                           % allow custom math operators
\usepackage{nicefrac}                          % compact symbols for 1/2, etc.
\usepackage{microtype}                         % microtypography
\usepackage{graphicx}                          % include images
\usepackage[ruled,linesnumbered]{algorithm2e}  % pseudocode
\usepackage[inline]{enumitem}                  % inline lists
\usepackage[capitalise]{cleveref}              % labelled references
\crefname{figure}{Figure}{Figures} 

\usepackage{fullpage}

\usepackage{bm}
\usepackage{bbm}
\usepackage{mathtools}
\usepackage{nccmath}
\usepackage{microtype}

\usepackage{amsthm}
\newtheorem{theorem}{Theorem}
\newtheorem{lemma}{Lemma}
\newtheorem{proposition}{Proposition}

\usepackage{multirow}
\usepackage[inline]{enumitem}
\usepackage{csquotes}

\graphicspath{ {./figures/} }

\title{Randomized Algorithms for Monotone Submodular Function Maximization on the Integer Lattice}

\author {
 Alberto Schiabel\textsuperscript{\rm 1},
 Vyacheslav Kungurtsev\textsuperscript{\rm 2},
 Jakub Mare\v{c}ek \textsuperscript{\rm 2}
}
\date{%
 \footnotesize  \textsuperscript{\rm 1} Department of Mathematics, University of Padova\\
 \textsuperscript{\rm 2} Department of Computer Science, Czech Technical University in Prague\\
 alberto.schiabel@studenti.unipd.it, \{vyacheslav.kungurtsev, jakub.marecek\}@fel.cvut.cz
}

\begin{document}
\maketitle              % typeset the header of the contribution

\begin{abstract}

Optimization problems with set submodular objective functions have many real-world applications.
In discrete scenarios, where the same item can be selected more than once, the domain of the target problem is generalized from a finite set to a bounded integer lattice.
In this work, we consider the problem of maximizing a monotone submodular function on the bounded integer lattice subject to a cardinality constraint.
In particular, we focus on maximizing DR-submodular functions, i.e., functions defined on the integer lattice that exhibit the  \emph{diminishing returns} property.
Given any $\varepsilon > 0$, we present a randomized algorithm with probabilistic guarantees of $\mathcal{O}(1 - \frac{1}{e} - \varepsilon)$ approximation, using a framework inspired by a
\textsc{Stochastic Greedy} algorithm developed for set submodular functions by Mirzasoleiman \emph{et al}.% \cite{mirzasoleiman:2015}.
% Most people frown upon citations in the abstract. \textsc{Stochastic Greedy} algorithm of Mirzasoleiman \emph{et al.} is ok.
We then show that, on synthetic DR-submodular functions, applying our proposed algorithm on the integer lattice is faster
than the alternatives, including reducing a target problem to the set domain and then applying the fastest known set submodular maximization algorithm.
%while retaining the same approximation ratio.
%Moreover, in expectation, our algorithms are more efficient than the two known fastest algorithms for the class of problems we study.
\end{abstract}

%\keywords{Submodular Function Maximization,and Integer Lattice \and Cardinality Constraint}

\section{Introduction}

\noindent In combinatorial optimization, machine learning, and operations research, Submodular function maximization problems are ubiquitous \cite{tohidi:2020}.
% We refer to \cite{tohidi:2020} for a recent tutorial and survey. 
Consider, for example, facility location \cite{cornuejols:1977} and sensor placement \cite{krause:2006} problems in operations research.
In machine learning, notable examples include 
experiment design \cite{agrawal:2019,sahin:2020},
dictionary learning \cite{krause:2010,das:2011}, 
and sparsity inducing regularizers \cite{bach:2010}.
In these problems, the goal is to pick a subset of a \emph{ground set} \hbox{$\mathcal{V} \coloneqq \{ 1, 2, \dots, n \}$}
that maximizes a set function $f : 2^{\mathcal{V}} \rightarrow \mathbb{R}$ defined on the powerset of $\mathcal{V}$. \\

\noindent Yet, there are practical applications in which one is not only interested in knowing whether an element $e \in \mathcal{V}$ is selected, but is also concerned with the amount of copies to select for a given $e$.
One such case is represented by the optimal budget allocation problem \cite{alon:2012}.
In these scenarios, the ground set can be considered as a multiset,
or equivalently as a cube $\{ \mathbf{x} \in \mathbb{Z}^{\mathcal{V}}_{+}\ |\ \mathbf{x} \preccurlyeq b \}$ on the integer lattice $\mathbb{Z}^{\mathcal{V}}_{+}$, where $\mathbf{b} \in \mathbb{Z}^{\mathcal{V}}_{+}$
is a known vector such that $b_e$ indicates the quantities available for each element $e \in \mathcal{V}$
and $\mathbf{x} \preccurlyeq \mathbf{y}$ means $x_e \leq y_e$ for every element $e \in \mathcal{V}$.
Any set function $f : 2^{\mathcal{V}} \rightarrow \mathbb{R}$ can be transformed into a pseudo-Boolean function
\hbox{$\phi : \{ 0, 1 \}^{\mathcal{V}} \rightarrow \mathbb{R}$}
defined on the Boolean lattice \hbox{$(\{ 0, 1 \}^{\mathcal{V}}, \land, \lor)$} \citep{crama:2011}, hence submodular optimization
performed on the integer lattice $\mathbb{Z}^{\mathcal{V}}$ can be seen as a natural generalization of optimization on the Boolean lattice.
There are, however, some important differences to highlight.  \\

\noindent We recall that a set function $f: 2^{\mathcal{V}} \rightarrow \mathbb{R}$ is called set \emph{submodular}
if and only if 
\begin{equation} \label{eq:set-submodularity}
f(A) + f(B) \geq f(A \cup B) + f(A \cap B) \quad \text{ for all } A, B \subseteq \mathcal{V}
\end{equation}
A function $f : \mathbb{Z}^{\mathcal{V}} \rightarrow \mathbb{R}$ is said to be \emph{integer-lattice submodular} if 
\begin{align}
    \label{eq:weak-dr-submodularity}
    f(\mathbf{x}) + f(\mathbf{y}) \geq f(\mathbf{x} \land \mathbf{y}) + f(\mathbf{x} \lor \mathbf{y})
\end{align}
for any $\mathbf{x}, \mathbf{y} \in \mathbb{Z}^{n}_{+}$.
On the other hand, $f : \mathbb{Z}^{\mathcal{V}} \rightarrow \mathbb{R}$ is called \emph{DR-submodular} if
\begin{align}
    \label{eq:strong-dr-submodularity-2}
    f(\mathbf{x} + \mathbf{1}_e) - f(\mathbf{x}) \geq f(\mathbf{y} + \mathbf{1}_e) - f(\mathbf{y}) 
\end{align}
for all $e \in \mathcal{V}$, and for all $\mathbf{x}, \mathbf{y} \in \mathbf{Z}^{n}_{+}$ such that $\mathbf{x} \preccurlyeq \mathbf{y}$.
Properties (\ref{eq:weak-dr-submodularity}) and (\ref{eq:strong-dr-submodularity-2}) are sometimes known as \emph{weak} and \emph{strong DR-submodularity}, respectively \citep{sahin:2020}. \\

\noindent In this work, we focus on \emph{monotone} submodular function maximization on the integer lattice subject to a cardinality constraint.
More precisely, consider the maximization problem:
\begin{equation}
    \label{eq:monotone-lattice-submax-under-cardinality-constraint}
    \begin{alignedat}{2}
    \max        & \quad f(\mathbf{x}) & \\
    \text{s.t.} & \quad ||\mathbf{x}||_1 \leq r & \\
                & \quad x_e \leq b_e & \quad \text{ for all } e \in \mathcal{V} \\
                & \quad \mathbf{x} \in \mathbb{Z}_{+}^{\mathcal{V}}, \mathbf{b} \in \mathbb{Z}_{+}^{\mathcal{V}}, r \in \mathbb{Z}_{+}
    \end{alignedat}
\end{equation}
where $f : \mathbb{Z}^{\mathcal{V}}_{+} \rightarrow \mathbb{R}$ is a monotone submodular function
defined on the integer lattice, $\mathbf{x} \in \mathbb{Z}^{\mathcal{V}}$ is defined such that $x_e \in \mathbb{Z}_{+}$ determines how many copies of an element $e \in \mathcal{V}$ should be selected,
$\mathbf{b} \in \mathbb{Z}^{\mathcal{V}}$ is a known vector where $b_e \in \mathbb{Z}_{+}$ represents how many copies of an element $e \in \mathcal{V}$ are available in the ground set, and $r \in \mathbb{Z}_{+}$ denotes the maximum cardinality of a feasible solution. We observe that
it should also hold that $r < ||\mathbf{b}||_1$, otherwise \cref{eq:monotone-lattice-submax-under-cardinality-constraint} becomes an unconstrained submodular maximization problem, which is trivially solved in constant time since $f$ is
monotone\footnote{In the monotone case of unconstrained submodular function maximization, selecting $\mathbf{x} \coloneqq \mathbf{b}$ guarantees that $f(\mathbf{x})$ is maximized.}. \\

\noindent We assume that $f : \mathbb{Z}^{\mathcal{V}} \rightarrow \mathbb{R}$ is given via a \emph{value oracle} black box, i.e., given some feasible $\mathbf{x} \in \mathbb{Z}^{\mathcal{V}}_{+}$, the oracle returns $f(\mathbf{x})$.
We refer to each oracle invocation as an \emph{oracle query}.
This value oracle model is standard in submodular optimization, as it abstracts away the details of the specific problem in a generic algorithm that solves a submodular optimization problem. \\
We also assume that $b_e \geq 1$ for all $e \in \mathcal{V}$
(otherwise, the elements $e' \in \mathcal{V}$ such that $b_{e'} = 0$ can be removed from the ground set $\mathcal{V}$ w.l.o.g.).
Moreover, notice that when $b_e = 1$ for all $e \in \mathcal{V}$, \cref{eq:monotone-lattice-submax-under-cardinality-constraint}
is equivalent to a set submodular function maximization problem subject to a cardinality constraint,
where we one can select a set $S \subseteq \mathcal{V}$ with at most $|S| = r$ elements. \\

\noindent Considering the wealth of algorithms for the optimization of set submodular functions, which are limited to express binary 
decisions (i.e., whether to select one element or not), it is natural to consider a \emph{reduction} from the integer lattice setting to the set submodular setting that enables the use of popular set submodular maximization procedures like the \textsc{Greedy} algorithm in the integer lattice domain.
The most natural one generates $b_e$ elements in a new set  $\mathcal{V}'$  for each element $e \in \mathcal{V}$ \cite{soma:2015}.
The drawback is that this reduction yields a pseudo-polynomial-time algorithm in $n$, which negatively affects the runtime of the set submodular algorithms as the value of $b_e$ grows for each $e$.
For DR-submodular functions, \cite{ene:2016} proposed another reduction to set-submodular optimization, 
which enacts a bit decomposition argument that yields a ground set $\mathcal{V}'$ of size $|\mathcal{V}'| = \mathcal{O}(\log \mathbf{b} + \frac{1}{\varepsilon}) \cdot n$. \\

\noindent The algorithm presented in this work is based on the \textsc{Stochastic Greedy} technique introduced by \cite{mirzasoleiman:2015}, which is -- in turn -- based on the \textsc{Greedy} algorithm \cite{nemhauser:1978}, which is perhaps the single most famous result in submodular function maximization.
Indeed, \cite{nemhauser:1978} showed that a simple greedy approach that picks
an unselected element at each iteration maximizing the local marginal gain provides a tight $(1 - \frac{1}{e})$-approximation guarantee for maximizing a \emph{monotone} submodular function subject to cardinality constraints. \\
\cite{mirzasoleiman:2015} devised an algorithm called \textsc{Stochastic Greedy} (\textsc{SG} for short), which greatly improves upon the running time of \textsc{Greedy} while retaining the same approximation ratio in expectation, using only $\mathcal{O}(n \log \frac{1}{\varepsilon})$ oracle queries.
The \textsc{SG} algorithm uses a randomized subsampling technique that, for each of the $r$ iterations, randomly selects a subset $Q \subseteq \mathcal{V}$ of a fixed size and finds the element $e \in Q$ that maximizes the marginal gain.
The key difference between this approach and \textsc{Greedy} is that $Q$ changes at each iteration. Thus, in expectation, \textsc{SG} does cover the entire dataset. In constrast, \textsc{Greedy} is equivalent to subsampling elements from $\mathcal{V}$ before looking for the most representative elements in it. 
Here, we extend the idea underlying the \textsc{Stochastic Greedy} algorithm from a set to an integer lattice domain,
proving that our algorithms are both practically and theoretically faster than \textsc{Stochastic Greedy} consider a reduction of an integer lattice problem to the set domain.

\subsection{Our Results}

Here, we present a randomized algorithm for maximizing a monotone \hbox{(DR-)submodular} function defined
on the integer lattice, improving upon the state of the art \citep{soma:2018,lai:2019} in terms of practical running time,
while obtaining a strong approximation guarantee in high probability.
We also show that our algorithms are significantly more stable than the previous state of the art in terms of the number of oracle queries required to approximately solve the target problem. \\

\noindent In particular, our contribution comprises:
\begin{itemize}
    \item a \textsc{Stochastic Greedy Lattice} (\textsc{SGL} for short) algorithm for maximizing monotone DR-submodular functions defined on the integer lattice subject to cardinality constraints;
    \item an analysis of the approximation ratio of \textsc{SGL}, which can be made arbitrarily close to  $(1 - \frac{1}{e})$, which is tight \citep{feige:1998}, with arbitrary probability greater or equal to one half;
    \item empirical experiments on a synthetic class of instances, which indicate the scalability of \textsc{SGL} and show the instability of the algorithms of \cite{soma:2018} and \cite{lai:2019}, which are considered the current state of the art for the considered constrained maximization problem.
\end{itemize}

\section{Background and Notation}

\paragraph{Notation}

We consider a finite $n$-dimensional set $\mathcal{V} \coloneqq \{ 1, 2, \dots, n \}$,
which we refer to as \emph{ground set}, and its powerset $2^{\mathcal{V}}$.
We denote by $\mathbb{R}_{+}$ the set of non-negative real numbers, and by $\mathbb{Z}_{+}$
the set of non-negative integer numbers.
We use bold face letters such as $\mathbf{x} \in \mathbb{R}^{\mathcal{V}}$ and $\mathbf{x} \in \mathbb{R}^n$ interchangeably to denote $n$-dimensional vectors. Similarly, we let $\mathbf{0}$ and $\mathbf{1}$ indicate $n$-dimensional vectors
whose values are all $0$ or $1$, respectively.
Given a vector $\mathbf{x}$, we denote its $e$-th coordinate by $x_e$ or $\mathbf{x}(e)$.
We let $\mathbf{1}_e \in \mathbb{Z}_{+}$ indicate the \emph{characteristic vector}, defined such that
$\mathbf{1}_e(e) \coloneqq 1$ and $\mathbf{1}_e(e') \coloneqq 0$ for all $e, e' \in \mathcal{V}$ such that $e \neq e'$.
We let $supp(\mathbf{x}) \coloneqq \{ e \in \mathcal{V}\ |\ x_e > 0 \}$ denote the \emph{support} of $\mathbf{x} \in \mathbb{Z}_{+}^{\mathcal{V}}$.
For two vectors $\mathbf{x}, \mathbf{y} \in \mathbb{R}^{n}_{+}$, $\mathbf{x} \preccurlyeq \mathbf{y}$ means $x_e \leq y_e$ for every element $e \in \mathcal{V}$.
Additionally, given $\mathbf{x}, \mathbf{y} \in \mathbb{Z}^{n}_{+}$ we let $\mathbf{x} \land \mathbf{y}$ and $\mathbf{x} \lor \mathbf{y}$
denote the coordinate-wise minimum and maximum, respectively,
i.e., $(\mathbf{x} \land \mathbf{y})_e \coloneqq \min \{ x_e, y_e \}$ and $(\mathbf{x} \lor \mathbf{y})_e \coloneqq \max \{ x_e, y_e \}$.
We define $||\mathbf{x}||_{1} \coloneqq (\sum_{e \in \mathcal{V}} |e_k|)$
and $||\mathbf{x}||_{\infty} \coloneqq max_{e \in \mathcal{V}} |x_e|$.

\paragraph{Definitions}

A set function $f : 2^{\mathcal{V}} \rightarrow \mathbb{R}$ is \emph{monotone} if, for every $A \subseteq B \subseteq \mathcal{V}$,
$f(A) \leq f(B)$. $f$ is said to be \emph{normalized} if $f(\varnothing) = 0$.
The marginal gain obtained by adding an element $e \in \mathcal{V}$ to a set $S \subseteq \mathcal{V}$ is defined as
$f(e\ |\ S) \coloneqq f(S \cup \{ e \}) - f(S)$. 
We now give the analogous definitions for the integer lattice case.

\noindent An integer lattice function $f : Z^{\mathcal{V}}_{+} \rightarrow \mathbb{R}$ is \emph{monotone} if $\mathbf{x} \preccurlyeq \mathbf{y}$ implies $f(\mathbf{x}) \leq f(\mathbf{y})$ for some $\mathbf{x}, \mathbf{y} \in \mathbb{Z}^{\mathcal{V}}_{+}$,
and $f$ is \emph{normalized} if $f(\mathbf{0}) = 0$.
We only consider monotone and normalized functions.
The marginal gain obtained by adding an element $e \in \mathcal{V}$ to a vector $\mathbf{x} \in Z^{\mathcal{V}}_{+}$ is defined as
$f(\mathbf{1}_e\ |\ \mathbf{x}) \coloneqq f(\mathbf{x} + \mathbf{1}_e) - f(\mathbf{x})$, where the sum operation is applied component-wise. 
 Given a known vector $\mathbf{b} \in Z^{\mathcal{V}}_{+}$ that defines the multiplicities of the elements $e \in \mathcal{V}$,
the problem of maximizing an integer lattice function $f : Z^{\mathcal{V}}_{+} \rightarrow \mathbb{R}$ under a cardinality constraint $r \in \mathbb{Z}_{+}$ is formalized as selecting a vector $\mathbf{x} \in Z^{\mathcal{V}}_{+}$ whose entries $x_e$ have value at most $b_e$ and such that $||\mathbf{x}||_1 \leq r$.

\subsection{Prior Work}

\noindent The problem of maximizing a set submodular function $f : 2^{\mathcal{V}} \rightarrow \mathbb{R}$,
i.e., finding a set $S \subseteq \mathcal{V}$ such that $f(S)$ is maximized, is NP-Hard, as it generalizes the \textsc{MaxCut} problem. 
%Since the brute-force enumeration of the $2^n$ possible subsets of $\mathcal{V}$ to select the maximizer $S$ is not feasible, this motivated a long history of research on greedy and local search algorithms with approximation guarantees.
Nemhauser \emph{et al.} first showed a discrete greedy algorithm (commonly referred to as \textsc{Greedy}) which yields a constant
$(1 - \frac{1}{e})$-approximation
to the problem of maximizing a monotone submodular function under a cardinality constraint \cite{nemhauser:1978}.
%The \textsc{Greedy} algorithm starts with the empty set $S_0 = \varnothing$, iteratively adding one element at a time for $k$ iterations. At each step, it picks the unselected element increasing the value of the current solution the most, i.e., the element with the largest marginal value w.r.t. the current solution.
The approximation ratio of \textsc{Greedy} is tight, i.e., it is the best possible performance guarantee for the considered problem
both in the value oracle model and independently of $P \neq NP$ \cite{feige:1998}.
%Other combinatorial approaches to solve submodular maximization subject to more generic matroid constraints include \cite{conforti:1984} and \cite{calinescu:2007}.
\cite{vondrak:2008} and \cite{calinescu:2011} pioneered continuous algorithms for constrained submodular maximization based on the \emph{multilinear extension} \cite{calinescu:2007} and the \emph{pipeage rounding} method \cite{ageev:2004, gandhi:2006}.
%\footnote{Given a normalized set submodular function $f : 2^{\mathcal{V}} \rightarrow \mathbb{R}$, its multilinear extension $F : [0, 1]^{\mathcal{V}} \rightarrow \mathbb{R}$
%agrees with $f$ on any vertex of the unit hypercube, and
%is such that is $F(\bm{x})$ is equivalent to the expected value of $f(R)$, where $R$ is a random set
%obtained by selecting each element $e \in \mathcal{V}$ independently with probability $x_e$.}
of $f$ \cite{calinescu:2007}. 
%It has later been adapted for non-monotone set functions as well with the \textsc{Measured Continuous Greedy} algorithm proposed by \cite{feldman:2011}.
While continuous algorithms are interesting on their own and can model a broad class of constraints, they are generally slower than the fastest greedy counterparts \cite{mirzasoleiman:2015}.  \\

\noindent Submodular optimization on the lattice domain is almost as old as set-submodular optimization itself \cite{topkis:1978}, but efficient algorithms for maximizing integer lattice submodular functions are much more recent.   
\cite{soma:2018} proposed two deterministic $(1 - 1/e - \varepsilon)$-approximation algorithms for maximizing monotone integer lattice submodular functions subject to cardinality constraints: one for the case of DR-submodular functions, and the other for the less restrictive case of integer lattice submodular functions (Algorithm 1 and Algorithm 3 of \cite{soma:2018}, respectively). We refer to these two algorithms as \textsc{Soma-I-DR} and \textsc{Soma-II}.
\textsc{Soma-I-DR} requires $\mathcal{O}(\frac{n}{\varepsilon} \mathbf{ ||\mathbf{b}||_{\infty} } \log{ \frac{r}{\varepsilon} })$ running time, whereas \textsc{Soma-II} has the much worse time complexity of $\mathcal{O}(\frac{n}{\varepsilon^2} \log{ ||\mathbf{b}||_{\infty} } \log{ \frac{r}{\varepsilon} } \log{\tau} )$, where $\tau$ is the ratio of the maximum value of $f$ to the minimum positive increase in the value of $f$.
\cite{gottschalk:2015} presented a natural \enquote{Double Greedy} time algorithm, which is also pseudopolynomial, and with the same approximation ratio of $(1 - 1/e)$ (and which is a $1/3$-approximation when there are no cardinality constraints). 
Finally, \cite{bach:2019} extended the multilinear extension to DR-submodularity, and \cite{sahin:2020b} introduced the generalized multilinear extension of the integer-lattice optimization problems in a similar spirit, but without translating these to algorithms.

%Other generalizations of set submodular functions, like \emph{bisubmodular} and \emph{k-submodular} functions have also been recently explored in the literature \cite{ward:2016}, leading to approximate maximization algorithms.
%\cite{kuhnle:2018,qian:2018} provided  approximation algorithms for maximizing a non-submodular function on the integer lattice subject to a cardinality constraint.
%\cite{maehara:2021} consider multiple knapsack constraints on the distributive lattice. 
%\cite{sahin:2020} considered the optimization of DR-submodular functions subject to discrete polymatroid constraints. 
%\cite{bian:2017a,bian:2017b} were the first to consider the continuous variants of DR-submodular functions, which has become a popular research direction \cite{niazadeh:2020,feldman:2020,mokhtari:2020}.

% As we have mentioned in the introduction,

\paragraph{Reduction from Lattice to Set}

\noindent We are aware of two generic reductions from the integer lattice domain to the set submodular setting, allowing the use of constrained maximization algorithms that target set submodular functions.
The most intuitive one \cite{soma:2015} generates $b_e$ copies for each element $e \in \mathcal{V}$, letting $\mathcal{V}'$ be the the multiset containing these copies.
Thus, problem (\ref{eq:monotone-lattice-submax-under-cardinality-constraint}) can then be reformulated into a set function maximization problem with respect to the ground set $\mathcal{V}'$, and one can use the \textsc{SG} algorithm to solve the problem. The drawback is that this reduction yields pseudo-polynomial time algorithm, since $|\mathcal{V}'| = \mathbf{b} \cdot \mathbf{1}n$.
We denote the \textsc{SG} algorithm simulated in the integer lattice domain by \textsc{SSG}.
Another, more efficient reduction algorithm is proposed by \cite{ene:2016}, but it can only be applied to DR-submodular functions.
This reduction uses a bit decomposition argument that yields a ground set $\mathcal{V}'$ of size $|\mathcal{V}'| = \mathcal{O}(\log \mathbf{b} + \frac{1}{\varepsilon}) \cdot n$.

\section{Our Algorithm}

Consider the drawbacks of the known reduction algorithms to translate an integer lattice submodular function to a set submodular function, explained above,
and the unpredictable running time of \textsc{Lai-DR} \citep{lai:2019}, which requires solving an integer linear program, whose runtime could be much larger than the time needed to evaluate an oracle query.
This motivated us to seek novel algorithms for submodular maximizing of integer lattice submodular functions subject to a cardinality constraint.
In particular, we were inspired by the simplicity and performance of the \textsc{Stochastic Greedy} algorithm \citep{mirzasoleiman:2015}, which employs the random subsampling technique,
while retaining a tight approximation ratio of $(1 - 1/e)$ on average for the set submodular version of the problem.
We thus propose a novel algorithm, which is based on the same random subsampling idea, but combined with the \textsc{Decreasing Threshold Greedy} framework of \cite{badanidiyuru:2014b}
and the \emph{binary search} strategy of \cite{soma:2018}.% whenever DR-submodularity holds. \\

\noindent Let \hbox{$s \coloneqq \lfloor \frac{n}{r} \cdot \log{\frac{1}{\varepsilon}} \rfloor$} indicate the sample size, as in \cite{mirzasoleiman:2015}.
\noindent Our proposed algorithm starts with an empty solution vector $\mathbf{x} \in \mathbb{Z}_{+}^{\mathcal{V}}$
and with a decreasing threshold value $\Theta$ (which we use to decide whether to add new elements to $\mathbf{x}$)
initialized with $d$, the maximum value of the function $f$ over every singleton $\mathbf{e}$, for $e \in \mathcal{V}$. \\
Until the cardinality of $\mathbf{x}$ reaches $r$, we randomly sample $s$ elements from the available elements in $\mathcal{V}$.
Locally, we look for $s$ distinct elements, but we allow sampling with replacement among different iterations.
We then look for the maximum integer $k \in \mathbb{Z}_{+}$ such that the marginal gain $f(k \cdot \mathbf{1}_e\ |\ \mathbf{x})$
is greater or equal than $k \cdot \Theta$.
The value of $k$ indicates how many copies of an element $e \in \mathcal{V}$ are added to $\mathbf{x}$ (if any at all),
so to make sure that we don't surpass the cardinality bound $r$,
at each iteration $t$ it should hold that $k \leq \min \{ b_e - x_{e}, r - ||\mathbf{x}||_1 \}$
for any previous solution $x^{t - 1}$ and element $e \in \mathcal{V}$.
If such value $k$ exists and satisfies the constraint that $f(\mathbf{x} + k \cdot \mathbf{1}_e)$
increase the value $f$, then we add $k$ copies of element $e$ to our solution $\mathbf{x}$.
We repeat the process until the sum of the entries of $\mathbf{x}$ reaches the desired cardinality $r$. \\

\noindent See \cref{alg:stochastic-greedy-lattice-iii} for a concrete instance of this approach, which we call \textsc{Stochastic Greedy Lattice} (\textsc{SGL} for short).
The major difference between our algorithm \textsc{SGL} and \hbox{\textsc{Soma-DR-I}} \cite{soma:2018} is that
at each iteration we only consider a subset $Q \subseteq \mathcal{V}$ of at most $s$ elements, whereas \hbox{\textsc{Soma-DR-I}} applies our same binary search procedure to every element in the ground set $\mathcal{V}$, and that \textsc{Soma-DR-I} is fully deterministic.
Moreover, we require $||\mathbf{x}||_1$ to reach the desired cardinality $r$, whereas \hbox{\textsc{Soma-DR-I}} only stops when the decreasing
threshold $\Theta$ is too low. On the other hand, we do not observe similarities between \textsc{SGL} and \textsc{Lai-DR},
except the fact that both algorithms are randomized.

\begin{algorithm}[tb]
    \caption{\textbf{SGL-III} algorithm for maximizing a monotone integer lattice function subject to cardinality constraint.}
    \label{alg:stochastic-greedy-lattice-iii}
    \SetKwFunction{FMain}{SGL-III}
    \SetKwProg{Fn}{Function}{:}{}
    \SetKwInOut{Input}{Input}
    \SetKwInOut{Output}{Output}
    \SetKw{Break}{break}
    \SetKw{And}{and}
    \Input{Monotone integer lattice function $f : \mathbb{Z}^{\mathcal{V}}_{+} \rightarrow \mathbb{R}$ (with $n \coloneqq |V|$), vector of quantities $\mathbf{b} \in \mathcal{Z}_{+}^{\mathcal{V}}$, and cardinality constraint $r \in \mathbb{Z}_{+}$.}
    %, error threshold $\varepsilon > 0$.}
    \Output{Vector $\mathbf{x} \in \mathbb{Z}^{\mathcal{V}}_{+}$ that approximately solves $\max_{\mathbf{x}} f(\mathbf{x})$ s.t. $||x||_1 \leq r$.} 
    % \DontPrintSemicolon
    % $s \gets \lfloor \frac{n}{r} \cdot \log{\frac{1}{\varepsilon}} \rfloor$\;
    $\mathbf{x}^{0} \gets \mathbf{0};$ $d \gets \max_{e \in \mathcal{V}} f(\mathbf{1}_e)$\;
    $\Theta^{0} \gets d;$ $\Theta_{stop} \gets \frac{\varepsilon}{r} \cdot d$; $t \gets 1$\;
    \While{$||\mathbf{x}^{t - 1}||_1 < r$}{
        $Q \gets $ sample $s$ distinct elements from $\{ e \in {\mathcal{V}}\ |\ x_e < b_e \}$\;
        \For{$e \in Q$}{
            $k^{t}_{max} \gets \min \{ b_e - x^{t - 1}_{e}, r - ||\mathbf{x}^{t - 1}||_1 \}$\;
            Find the max $k^{t} \in \mathbb{Z}_{+} $ with $1 \leq k^{t} \leq k^{t}_{max}$
            such that $f(k^{t} \cdot \mathbf{1}_e\ |\ \mathbf{x}) \geq k^{t} \cdot \Theta^{t - 1}$ and
            $f(\mathbf{x}^{t - 1} + k^{t} \cdot \mathbf{1}_e) \geq f(\mathbf{x}^{t - 1})$ via binary search\;
            \If{$\exists\ k^{t}$ as above \And $f(\mathbf{x}^{t - 1} + k^{t} \cdot \mathbf{1}_e) \geq f(\mathbf{x}^{t - 1})$}{
                $\mathbf{x}^{t} \gets \mathbf{x}^{t - 1} + k^{t} \cdot \mathbf{1}_e$\;
            }
        }
        $\Theta^{t} \gets \max \{ \Theta^{t - 1} \cdot (1 - \varepsilon \cdot), \Theta_{stop} \}$; $t \gets t + 1$\;
    }
    \KwRet $\mathbf{x}^{t - 1}$\;
\end{algorithm}

\section{An Analysis}

\noindent We now develop guarantees for the proposed algorithm (\textsc{SGL}), both in terms of its approximation ratio and computational complexity expressed in terms of oracle queries.

\subsection{Probabilistic Approximation Guarantees}

We define the following set $A_t$ corresponding to the random subsampling procedure of $s$ elements at iteration $t$.
\begin{equation}\label{eq:defai}
% TODO: re-define A_t using a more understandable notation
A_t \coloneqq \text{supp}(\max\{\mathbf{x}^*-\mathbf{x}_t,0\}).
\end{equation}
\noindent We also define
\begin{equation}\label{eq:deftbound}
  \bar{t} \coloneqq \frac{\log\left[1-\exp\{-\frac{\log 2}{n}\}\right]}{\log\left(1-\frac{s}{n}\right)},
\end{equation}
and we prove the following:

\begin{lemma}\label{lem:somagain}
Consider $\mathbf{x}^{t}$ the current solution at iteration $t$,
and let $k$ be quantity corresponding to a selected element $e \in \mathcal{V}$ that should be added to $\mathbf{x}^{t}$ in \cref{alg:stochastic-greedy-lattice-iii}.
Then, with probability $p > \frac{1}{2}$, it holds that
\begin{equation}\label{eq:somagain}
\frac{f(k^t \cdot \mathbf{1}_e\ |\ \mathbf{x}^{t})}{k^t} \geq \frac{(1 - \bar t\epsilon)}{r} \sum_{v\in A_t \setminus \{e\}} f(\mathbf{1}_v\ |\ \mathbf{x}^t) .
\end{equation}
\end{lemma}

\begin{proof}
    We can see that \cref{eq:somagain} holds for $\Theta = d$ by the same arguments as ~\cite[Lemma 3]{soma:2018}, which require DR-submodularity.
    Let us suppose that it also holds for iteration $t > 0$.
    Indeed, given $v\in A_t\setminus \{e\}$, if $v$ was picked from the previous sample $Q_{t - 1}$ in quantity $k'$, we would have that
    $f(\mathbf{1}_v| \mathbf{x}^t) \leq \frac{\Theta}{1-\varepsilon}$.
    This is by contradiction, since if it was otherwise,
    \begin{align*}
        f((k'+1) \mathbf{1}_v\ |\ \mathbf{x}^{t-1}) \geq f(\mathbf{1}_v\ |\ \mathbf{x}^{t}) +
          f(k' \mathbf{1}_v\ |\ \mathbf{x}^{t-1}) > \frac{\Theta}{1-\varepsilon}+ \frac{k'\Theta}{1-\varepsilon}
    \end{align*}
    would hold, contradicting $k'$ being the largest feasible integer found at the previous iteration $t - 1$.
    Thus, for such a $v$, it holds that $\frac{f(k\mathbf{1}_e\ |\ \mathbf{x}^t)}{k} \geq (1 - \varepsilon) f(\mathbf{1}_v\ |\ \mathbf{x}^t)$. \\
    \noindent On the other hand, if $v \notin Q_{t - 1}$, then
    \begin{align*}
        \frac{f(k \mathbf{1}_e\ |\ \mathbf{x}^t)}{k} \geq (1 - \varepsilon)^{\hat{t(v,t)}} f(\mathbf{1}_v\ |\ \mathbf{x}^t)
          \geq (1 - \hat t(v,t) \varepsilon) f(\mathbf{1}_v\ |\ \mathbf{x}^t),
    \end{align*}
    where $\hat{t}(v, t)$ denotes the number of iterations which occurred before $t$ since the last time the element $v$ was selected.
    With a slight abuse of notation, we fix $v$ and $t$ and define $\hat{t} \coloneqq \hat{t}(v, t)$.
    Clearly, $\hat{t}$ is a geometric random variable with \emph{success probability} $\hat{p} \coloneqq \frac{s}{n}$
    and with \emph{cumulative distribution function} $1-(1-\hat{p})^{\hat{t}}$.
    Thus, for every element $v \in \mathcal{V}$ at most $\hat{t}$ iterations old,
    this would have a probability of $\left(1 - (1 - \hat{p})^{\hat{t}} \right)^{n}$. % TODO: what does "this" refer to exactly?
    Now,
    \begin{align*}
        \left(1 - (1 - \hat{p})^{\hat{t}} \right)^{n} > \frac{1}{2} \quad \Longleftrightarrow
          \quad \hat{t} < \bar{t} = \frac{\log\left[1-\exp\{-\frac{\log 2}{n}\}\right]}{\log\left(1-\frac{s}{n}\right)}
    \end{align*}
    Hence, with probability $p = \frac{1}{2}$ and for all iterations $t$ and elements $v \in \mathcal{V}$, it holds that
    \begin{align*}
          \frac{f(k^t \mathbf{1}_e\ |\ \mathbf{x}^t)}{k^t} \geq (1 - \varepsilon)^{\bar{t}} f(\mathbf{1}_v\ |\ \mathbf{x}^t) \geq (1 - \bar{t} \varepsilon) f(\mathbf{1}_v\ |\ \mathbf{x}^t).
    \end{align*}
    If we average over the elements, we obtain that
    \begin{align*}
        \frac{f(k^t \mathbf{1}_e\ |\ \mathbf{x}^t)}{k^t} \geq \frac{(1 - \bar{t} \varepsilon)}{r}
          \sum_{v\in A_t \setminus \{e\}} f(\mathbf{1}_v\ |\ \mathbf{x}^t ).
    \end{align*}
\end{proof}

\noindent \cref{lem:somagain} allows us to prove the following approximation bound. 
\begin{theorem}\label{th:approx-SGL-III}
\cref{alg:stochastic-greedy-lattice-iii} achieves an approximation ratio of $(1 - \frac{1}{e} - \bar{t} \varepsilon)$
with probability $p > \frac{1}{2}$.
\end{theorem}

\begin{proof}
By \cref{lem:somagain}, with probability $p > \frac{1}{2}$ we have that
\begin{align*}
    \frac{f(k^t \mathbf{1}_e\ |\ \mathbf{x}^t)}{k^t} \geq \frac{(1 - \bar{t} \varepsilon)}{r} f(\mathbf{x}^{*} - \mathbf{x}^{t}\ |\ \mathbf{x}^t),
\end{align*}
where $\mathbf{x}^t$ is the solution of \cref{alg:stochastic-greedy-lattice-iii} at iteration $t$, and $\mathbf{x}^{*}$ is the optimal solution.
Given $k_{sum}$ the sum of all $k$ selected for a fixed element $e \in \mathcal{V}$ at iteration $t$,
it follows that
\begin{align*}
    f(\mathbf{x}^{t + 1}) - f(\mathbf{x}^t) \geq \frac{(1 - \bar{t} \varepsilon) k_{sum}}{r} f(\mathbf{x}^*) - f(\mathbf{x}^t).
\end{align*}
Let $\mathbf{x}^{t_f}$ be the result of \cref{alg:stochastic-greedy-lattice-iii}.
By induction,
\begin{align*}
    \begin{array}{l}
    f(\mathbf{x}^{t_f}) \geq \left(1 - \prod_t \left(1-\frac{(1 - \bar{t} \varepsilon) k_{sum}}{r} \right)\right) f(\mathbf{x}^{*})  \\
    \qquad \geq \left(1 - \prod_t \exp(-\frac{(1 - \bar{t}\varepsilon)k_{sum}}{r})\right) f(\mathbf{x}^{*})
    = \left(1 - \exp(-\frac{(1 - \bar{t}\varepsilon)\sum_t k_{sum}}{r})\right) f(\mathbf{x}^{*})\\
    \qquad \geq (1 - \frac{1}{e} - \bar{t} \varepsilon) f(\mathbf{x}^{*}).
    \end{array}
\end{align*}
\end{proof}

\vskip 6mm
Notice that the probability $p$ in Lemma \ref{lem:somagain} can be adjusted by probability amplification \cite[Section 6.8]{motwani1995randomized}. In effect, Algorithm \textsc{SGL-III} is likely to achieve an arbitrarily close performance in terms of the approximation ratio as the state of the art~\cite{soma:2018}, but with a lower overall computational cost.

\subsection{Bounding the Number of Oracle Calls}

\begin{proposition}\label{lem:max-k-max}
    The largest possible value of $k^{t}_{max}$ is $\min \{ ||\mathbf{b}||_{\infty}, r \}$.
\end{proposition}

\begin{proof}
    In \cref{alg:stochastic-greedy-lattice-iii}, $k^{t}_{max}$ is computed as $\min \{ b_e - x^{t - 1}_{e}, r - ||\mathbf{x}^{t - 1}||_1 \}$
    for a given $e \in Q$.
    We recall that $\mathbf{b}$ and $r$ are two given constants, and that $||\mathbf{x}||_1$ is non-decreasing over time (starting from $\mathbf{x} = \mathbf{0}$). Thus, in the worst case, $k^{t}_{max} = \min \{ b_e - 0, r - 0 \} = \min \{ ||\mathbf{b}||_{\infty}, r \}$,
    which can only happen before the first element is selected to be part of the solution.
\end{proof}

\begin{proposition}\label{prop:binary-search-time}
    The binary search procedure requires \hbox{$\lceil \log_{2} (k^{t}_{max} + 1) \rceil$} oracle queries for a given $k^{t}_{max}$.
    The number of oracle queries performed by one iteration of \cref{alg:stochastic-greedy-lattice-iii} is $\mathcal{O}(\log(\min \{ ||\mathbf{b}||_{\infty}, r \}))$.
\end{proposition}

\begin{proof}
    At each iteration, we consider at most \hbox{$s \coloneqq \lfloor \frac{n}{r} \cdot \log{\frac{1}{\varepsilon}} \rfloor$} distinct elements.
    From \cref{lem:max-k-max}, we have at most $\mathcal{O}(\log(\min \{ ||\mathbf{b}||_{\infty}, r \}))$ oracle queries per iteration.
\end{proof}

\noindent The number of iterations is a function of a number of parameters, including 
the function $f$ and its $n \coloneqq |V|$ and the resulting $\Theta^{0} = d$, as well as $\mathbf{b}$, $r$, $\varepsilon$, and the associated $\Theta_{stop} = \frac{\varepsilon}{r} \cdot d$. 
We present details of the runtime of several variants of the algorithm in an extended version on-line. 
Let us illustrate the behaviour computationally.

\section{Numerical Results}

We have implemented our algorithm \textsc{SGL}, as well as thealgorithms of:
\begin{itemize}
    \item \textsc{Soma-DR-I} by Soma and Yoshida \cite{soma:2018};
    \item \textsc{Lai-DR} by Lai \cite{lai:2019};
    \item \textsc{SSG}, i.e., the Simulated \textsc{Stochastic Greedy} algorithm due to Mirzasoleiman \emph{et al.} \cite{mirzasoleiman:2015} lifted from the set to the integer lattice domain.
\end{itemize}
in an open-source package available on-line\footnote{\url{https://github.com/jkomyno/lattice-submodular-maximization}}. As far as we know, the code we release is the first open-source package for submodular function maximization on the integer lattice subject to cardinality constraints.

We have validated our proposed algorithm on some synthetic problem instances of increasing size $n$, cardinality constraint $r$, and for various uniformly random configurations of the vector of available quantities $\mathbf{b}$. We measure the performance of the algorithms both in terms of the number of oracle queries and in terms of the total running time.

\paragraph{Synthetic Monotone Function}

\noindent We define a simple monotone DR-submodular function $f : \mathbb{Z}^{\mathcal{V}}_{+} \rightarrow \mathbb{R}_{+}$ where $f(\mathbf{x}) \coloneqq \mathbf{w}^{\top} \cdot \mathbf{x}$,
wherein $w \in \mathbb{Z}^{\mathcal{V}}_{+}$ is a vector sorted in ascending order whose components $w_e$
are such that $1 \leq w_e \leq 100$, for every $e \in \mathcal{V}$. One can observe that $f$ is normalized.
We consider the following parameter settings:

\begin{itemize}
    \item $n \in \{ 100, 200, 500, 750 \}$;
    \item $r \in \{ 0.25 \cdot n, 0.5 \cdot n, 1 \cdot n, 2 \cdot n \}$;
    \item Six equidistant $b \in \mathbb{Z}_{+} \cap \left[ \lfloor \frac{r}{20} \rfloor, \lfloor \frac{r}{2} \rfloor \right]$
\end{itemize}
$\mathbf{b}$ is uniformly sampled in range $\{ b, b \cdot 4 \}$.
We discard every unfeasible combination of $(n, \mathbf{b}, r)$ such that $r > ||\mathbf{b}||_{\infty}$ and $r \geq n \cdot ||\mathbf{b}||_{\infty}$. \\

\paragraph{Experimental Setting}
\noindent We remark that every algorithm under consideration uses randomization, except \textsc{Soma-DR-I}.
To reduce the bias of our benchmark, for each parameter setting, we repeated every experiment on a randomized algorithm 5 times. Moreover, we seeded our random generator to ensure that
every algorithm consumes exactly the same problem instances, even though those are generated randomly \enquote{on the fly}. We fixed the timeout for the experiments at 6 hours and 30 minutes.
While \textsc{SGL} and \textsc{SSG} managed to complete their computations in time, both \textsc{Soma-DR-I} and \textsc{Lai-DR} (currently considered the state of the art for our target problem) timed out on the most challenging parameter settings, with $n \geq 500$.
% We conducted our experiments on a Linux server cluster equipped with two Intel Xeon CPUs
% and 384 GBs of RAM.
The code for the experiments is written in \textsc{Python 3.8}.
For all algorithms that require an error treshols $\varepsilon > 0$ in input (our included),
we fix $\varepsilon \coloneqq \frac{1}{4n}$, where $n$ is the size of the ground set. \\

%\subsection{Experimental Outcomes: Oracle Queries}

%Perhaps unsurprisingly, using algorithms for set submodular maximization on an \emph{expanded} ground set that mimics a multiset is not a scalable approach, even when using very fast algorithms like \textsc{Stochastic Greedy}.
% The number of oracle queries of \textsc{SSG} tend to explode towards 1 million the higher the values of $n, ||\mathbf{b}||_{\infty}$, and $r$.
% Another curious finding is that if we do the contrary, i.e., take a DR-submodular algorithm and set $\mathbf{b} \coloneqq \mathbf{1}$, \textsc{SSG} is the fastest algorithm, followed by \textsc{SGL}. \textsc{Soma-DR-I} tends to perform worst when $||\mathbf{b}||_{\infty}$ is low (i.e., between $1$ and $8$).

%\noindent We have noticed that the oracle query performance of \textsc{Soma-DR-I} is
%dif and only ificult to foresee, as it is extremely volatile.
%\textsc{Lai-DR} is slightly more stable, even though it suffers from spikes when the heuristics used by its ILP solver (we used GLPK via the cvxopt Python library) get stuck (TODO: rephrase better).

%\subsection{Experimental Outcomes: Approximated Value}

% Almost every one among the considered algorithms provide consistently high result values.

\begin{figure}[!tbh]
  \centering
  \includegraphics[scale=0.3]{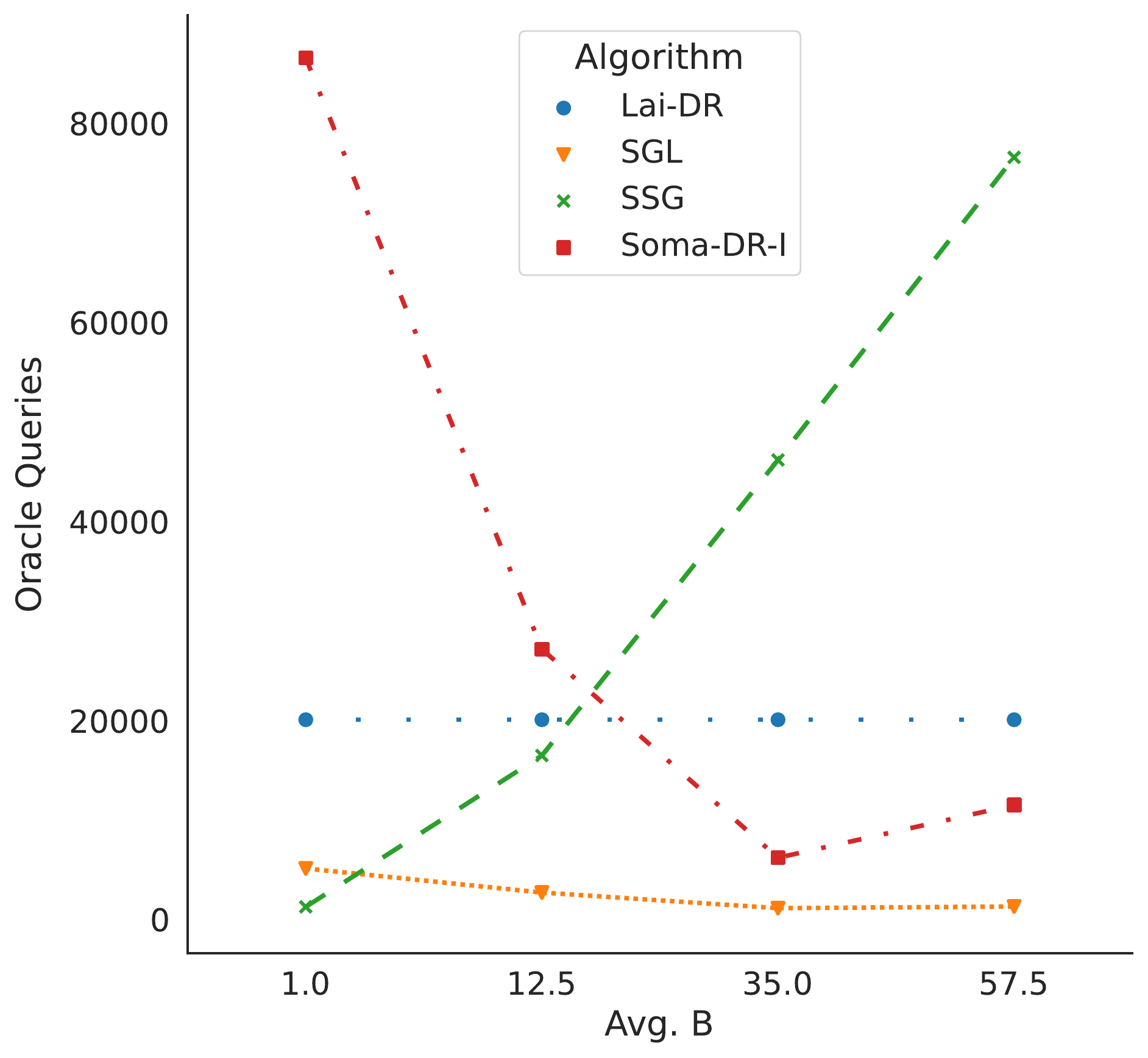}
  \centering
  \includegraphics[scale=0.3]{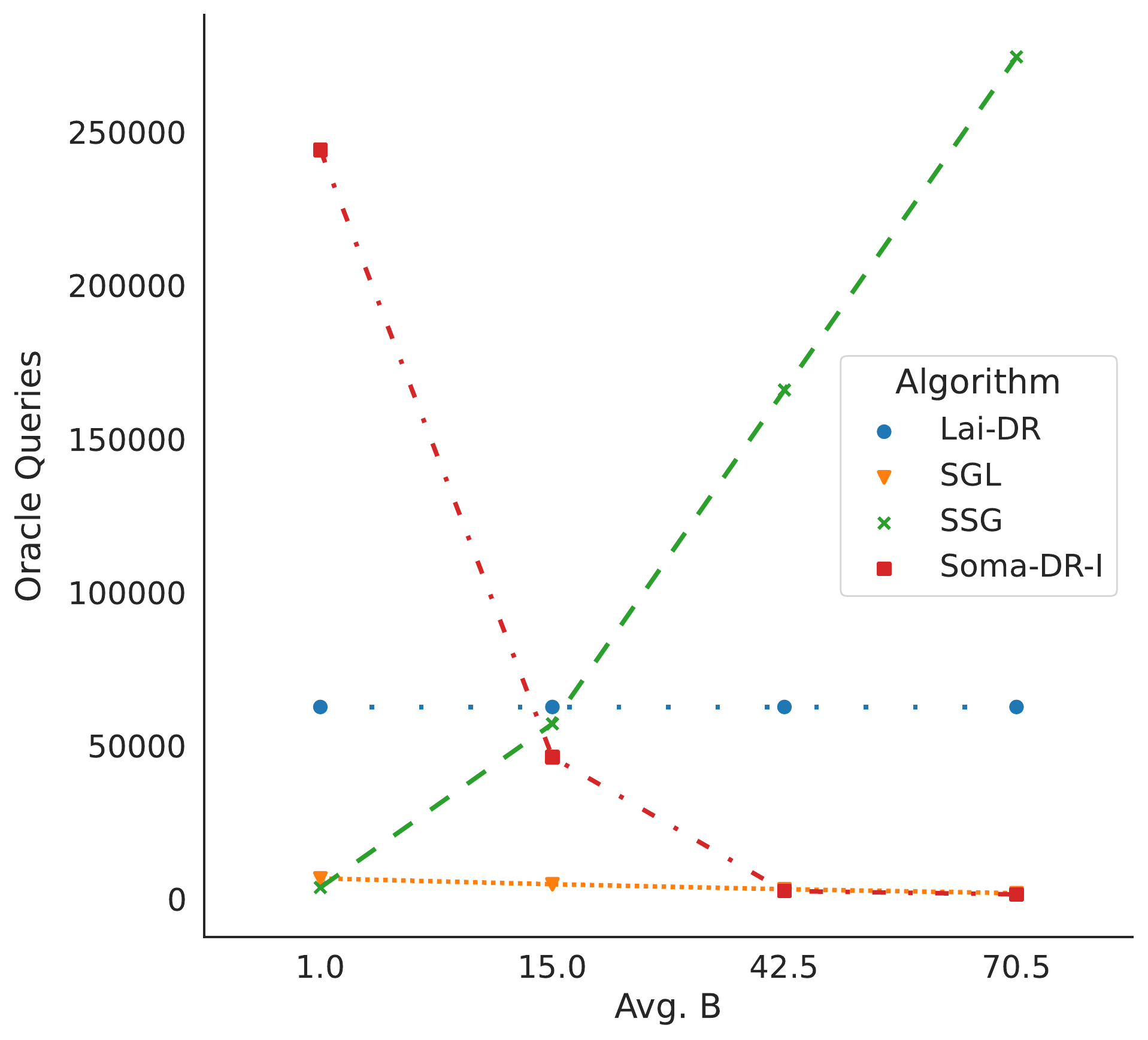}
\caption{An illustration of the growth of the number of oracle calls as a function of the average value of $\mathbf{b}$ on a synthetic DR-submodular monotone function, for all the considered algorithms. On the left, $n=200$ and $r=100$. On the right, $n=500$ and $r=125$.
It is clear that \textsc{SSG} depends positevely on the $\mathbf{b}$, whereas \textsc{Soma-DR-I} seem to present an inverse relationship (but not always).
\textsc{SGL} and \textsc{Lai-DR} do not seem to be affected by $\mathbf{b}$.
\textsc{SGL} is the only algorithm that consistently scales as both $n$ and $\mathbf{b}$ grow.}
\label{fig:demo-monotone-n_calls-by-algo-n_100-r_100}
\end{figure}

\begin{table}
\centering
\begin{tabularx}{\textwidth}{ l *{5}{Y} }
{} & {} & \multicolumn{4}{c}{Avg. Oracle Queries} \\
& n &           100 &           200 &            500 &            750 \\ \hline
Algorithm     &               &               &                &                \\
\midrule
\textsc{SGL} (this work)   & &   \textbf{1030.64} &   \textbf{1965.26} &    \textbf{3808.13} &    \textbf{5396.14} \\
\textsc{Lai-DR} \cite{lai:2019}    & &   8657.14 &  34457.14 &   62625.00 &          - \\
\textsc{SSG} \cite{mirzasoleiman:2015}       & &  16614.29 &  73978.57 &  531446.43 &  558141.00 \\
\textsc{Soma-DR-I} \cite{soma:2018} & &   7459.93 &  22261.79 &   72338.23 &          - \\
\end{tabularx}
\vspace{0.5cm}
\caption{Average number of oracle queries required by each considered algorithm as $n$ increases on the synthetic monotone DR-submodular instances. Smaller numbers are better. We see that \textsc{SGL} takes the least amount of time on average, and that the \emph{reduction from integer lattice to set domain} approach does not scale well.}
\label{tab:demo_monotone-n-vs-n_calls-pivot}
\end{table}

\begin{table}
\centering
\begin{tabularx}{\textwidth}{ l *{4}{Y} }
{} & {} & \multicolumn{3}{c}{Avg. Return Value} \\
& n &           100 &           200 &            500 \\ \hline
Algorithm     &               &               &            \\
\midrule
\textsc{SGL} (this work)               & &  8032.02          & 16262.60          &  40914.00 \\
\textsc{Lai-DR} \cite{lai:2019}        & &  4234.97          &  8528.03          &  6217.33  \\
\textsc{SSG} \cite{mirzasoleiman:2015} & &  8199.37          & \textbf{16625.77} &  \textbf{41854.79} \\
\textsc{Soma-DR-I} \cite{soma:2018}    & &  \textbf{8208.36} & 16615.21          & 	37472.54 \\
\end{tabularx}
\vspace{0.5cm}
\caption{Average value returned by each considered algorithm as $n$ increases on the synthetic monotone DR-submodular instances. Higher numbers are preferred. \textsc{SSG} tops the ranking in terms of value precision, with \textsc{SGL} coming close, less than $3\%$ far. Surprisingly, \textsc{Lai-DR}'s approximation performs much worse than expected compared to the other algorithms.}
\label{tab:demo_monotone-n-vs-n_calls-pivot}
\end{table}

\section{Conclusion and Future Work}

In this extended abstract, we considered the problem of maximizing a monotone integer lattice and DR-submodular functions subject to cardinality constraints and proposed one randomized algorithm for this problem, inspired by the random sampling technique that, until now, has only been applied to set submodular functions \cite{mirzasoleiman:2015}. We showed that the algorithm achieves state-of-of-the-art approximation ratio for monotone DR-submodular functions.
Experimentally, we have shown that the algorithm requires
considerably fewer value-oracle queries than state-of-the-art deterministic algorithms for the same problem, as well as the naive baseline lifting the integer lattice to a set domain.

%We are currently exploring other alternative algorithms that combine the random subsampling idea of \cite{mirzasoleiman:2015} to the integer lattice domain. In particular, we want to limit the number of iterations of \textsc{SGL} without losing in terms of approximation guarantee.

%\section*{Acknowledgments}
%V. K. and J. M. gratefully acknowledge funding from OP VVV project number CZ.02.1.01/0.0/0.0/16\_019/0000765 (``Research Center for Informatics'').

\bibliographystyle{IEEEtran}
\bibliography{main.bib}

\begin{thebibliography}{10}
\providecommand{\url}[1]{#1}
\csname url@samestyle\endcsname
\providecommand{\newblock}{\relax}
\providecommand{\bibinfo}[2]{#2}
\providecommand{\BIBentrySTDinterwordspacing}{\spaceskip=0pt\relax}
\providecommand{\BIBentryALTinterwordstretchfactor}{4}
\providecommand{\BIBentryALTinterwordspacing}{\spaceskip=\fontdimen2\font plus
\BIBentryALTinterwordstretchfactor\fontdimen3\font minus
  \fontdimen4\font\relax}
\providecommand{\BIBforeignlanguage}[2]{{%
\expandafter\ifx\csname l@#1\endcsname\relax
\typeout{** WARNING: IEEEtran.bst: No hyphenation pattern has been}%
\typeout{** loaded for the language `#1'. Using the pattern for}%
\typeout{** the default language instead.}%
\else
\language=\csname l@#1\endcsname
\fi
#2}}
\providecommand{\BIBdecl}{\relax}
\BIBdecl

\bibitem{tohidi:2020}
E.~Tohidi, R.~Amiri, M.~Coutino, D.~Gesbert, G.~Leus, and A.~Karbasi,
  ``Submodularity in action: From machine learning to signal processing
  applications,'' \emph{IEEE Signal Processing Magazine}, vol.~37, no.~5, pp.
  120--133, 2020.

\bibitem{cornuejols:1977}
G.~Cornuejols, M.~Fisher, and G.~L. Nemhauser, ``On the uncapacitated location
  problem,'' in \emph{Studies in Integer Programming}, ser. Annals of Discrete
  Mathematics, P.~Hammer, E.~Johnson, B.~Korte, and G.~Nemhauser, Eds.\hskip
  1em plus 0.5em minus 0.4em\relax Elsevier, 1977, vol.~1, pp. 163--177.

\bibitem{krause:2006}
A.~Krause, C.~Guestrin, A.~Gupta, and J.~Kleinberg, ``Near-optimal sensor
  placements: Maximizing information while minimizing communication cost,'' in
  \emph{Proceedings of the 5th international conference on Information
  processing in sensor networks}, 2006, pp. 2--10.

\bibitem{agrawal:2019}
R.~Agrawal, C.~Squires, K.~Yang, K.~Shanmugam, and C.~Uhler, ``Abcd-strategy:
  Budgeted experimental design for targeted causal structure discovery,'' in
  \emph{The 22nd International Conference on Artificial Intelligence and
  Statistics}.\hskip 1em plus 0.5em minus 0.4em\relax PMLR, 2019, pp.
  3400--3409.

\bibitem{sahin:2020}
A.~Sahin, J.~Buhmann, and A.~Krause, ``Constrained maximization of lattice
  submodular functions,'' in \emph{ICML 2020 workshop on Negative Dependence
  and Submodularity for ML, Vienna, Austria, PMLR 119, 2020.}, 2020.

\bibitem{krause:2010}
A.~Krause and V.~Cevher, ``Submodular dictionary selection for sparse
  representation,'' in \emph{Proceedings of the 27th International Conference
  on International Conference on Machine Learning}, ser. ICML'10.\hskip 1em
  plus 0.5em minus 0.4em\relax Madison, WI, USA: Omnipress, 2010, p. 567–574.

\bibitem{das:2011}
A.~Das and D.~Kempe, ``Submodular meets spectral: Greedy algorithms for subset
  selection, sparse approximation and dictionary selection,'' in
  \emph{Proceedings of the 28th International Conference on International
  Conference on Machine Learning}, ser. ICML'11.\hskip 1em plus 0.5em minus
  0.4em\relax Madison, WI, USA: Omnipress, 2011, p. 1057–1064.

\bibitem{bach:2010}
\BIBentryALTinterwordspacing
F.~Bach, ``Structured sparsity-inducing norms through submodular functions,''
  in \emph{Advances in Neural Information Processing Systems}, J.~Lafferty,
  C.~Williams, J.~Shawe-Taylor, R.~Zemel, and A.~Culotta, Eds., vol.~23.\hskip
  1em plus 0.5em minus 0.4em\relax Curran Associates, Inc., 2010. [Online].
  Available:
  \url{https://proceedings.neurips.cc/paper/2010/file/4b0a59ddf11c58e7446c9df0da541a84-Paper.pdf}
\BIBentrySTDinterwordspacing

\bibitem{alon:2012}
\BIBentryALTinterwordspacing
N.~Alon, I.~Gamzu, and M.~Tennenholtz, ``Optimizing budget allocation among
  channels and influencers,'' in \emph{Proceedings of the 21st International
  Conference on World Wide Web}, ser. WWW '12.\hskip 1em plus 0.5em minus
  0.4em\relax New York, NY, USA: Association for Computing Machinery, 2012, p.
  381–388. [Online]. Available: \url{https://doi.org/10.1145/2187836.2187888}
\BIBentrySTDinterwordspacing

\bibitem{crama:2011}
Y.~Crama and P.~L. Hammer, ``Boolean functions - theory, algorithms, and
  applications,'' in \emph{Encyclopedia of mathematics and its applications},
  2011.

\bibitem{soma:2014}
T.~Soma, N.~Kakimura, K.~Inaba, and K.-i. Kawarabayashi, ``Optimal budget
  allocation: Theoretical guarantee and efficient algorithm,'' in
  \emph{Proceedings of the 31st International Conference on International
  Conference on Machine Learning - Volume 32}, ser. ICML'14.\hskip 1em plus
  0.5em minus 0.4em\relax JMLR.org, 2014, p. I–351–I–359.

\bibitem{soma:2015}
\BIBentryALTinterwordspacing
T.~Soma and Y.~Yoshida, ``A generalization of submodular cover via the
  diminishing return property on the integer lattice,'' in \emph{Advances in
  Neural Information Processing Systems}, C.~Cortes, N.~Lawrence, D.~Lee,
  M.~Sugiyama, and R.~Garnett, Eds., vol.~28.\hskip 1em plus 0.5em minus
  0.4em\relax Curran Associates, Inc., 2015. [Online]. Available:
  \url{https://proceedings.neurips.cc/paper/2015/file/7bcdf75ad237b8e02e301f4091fb6bc8-Paper.pdf}
\BIBentrySTDinterwordspacing

\bibitem{ene:2016}
\BIBentryALTinterwordspacing
A.~Ene and H.~L. Nguyen, ``A reduction for optimizing lattice submodular
  functions with diminishing returns,'' \emph{CoRR}, vol. abs/1606.08362, 2016.
  [Online]. Available: \url{http://arxiv.org/abs/1606.08362}
\BIBentrySTDinterwordspacing

\bibitem{mirzasoleiman:2015}
B.~Mirzasoleiman, A.~Badanidiyuru, A.~Karbasi, J.~Vondr\'{a}k, and A.~Krause,
  ``Lazier than lazy greedy,'' in \emph{Proceedings of the Twenty-Ninth AAAI
  Conference on Artificial Intelligence}, ser. AAAI'15.\hskip 1em plus 0.5em
  minus 0.4em\relax AAAI Press, 2015, p. 1812–1818.

\bibitem{nemhauser:1978}
\BIBentryALTinterwordspacing
G.~L. Nemhauser, L.~A. Wolsey, and M.~L. Fisher, ``An analysis of
  approximations for maximizing submodular set functions{\textemdash}i,''
  \emph{Mathematical Programming}, vol.~14, no.~1, pp. 265--294, Dec. 1978.
  [Online]. Available: \url{https://doi.org/10.1007/bf01588971}
\BIBentrySTDinterwordspacing

\bibitem{soma:2018}
T.~Soma and Y.~Yoshida, ``Maximizing monotone submodular functions over the
  integer lattice,'' \emph{Mathematical Programming}, vol. 172, no. 1–2, p.
  539–563, Nov. 2018.

\bibitem{lai:2019}
\BIBentryALTinterwordspacing
L.~Lai, Q.~Ni, C.~Lu, C.~Huang, and W.~Wu, ``Monotone submodular maximization
  over the bounded integer lattice with cardinality constraints,''
  \emph{Discrete Mathematics, Algorithms and Applications}, vol.~11, no.~06, p.
  1950075, Dec. 2019. [Online]. Available:
  \url{https://doi.org/10.1142/s1793830919500757}
\BIBentrySTDinterwordspacing

\bibitem{feige:1998}
\BIBentryALTinterwordspacing
U.~Feige, ``A threshold of $ln(n)$ for approximating set cover,'' \emph{J.
  ACM}, vol.~45, no.~4, p. 634–652, 1998. [Online]. Available:
  \url{https://doi.org/10.1145/285055.285059}
\BIBentrySTDinterwordspacing

\bibitem{vondrak:2008}
\BIBentryALTinterwordspacing
J.~Vondrak, ``Optimal approximation for the submodular welfare problem in the
  value oracle model,'' in \emph{Proceedings of the Fortieth Annual ACM
  Symposium on Theory of Computing}, ser. STOC '08.\hskip 1em plus 0.5em minus
  0.4em\relax New York, NY, USA: Association for Computing Machinery, 2008, p.
  67–74. [Online]. Available: \url{https://doi.org/10.1145/1374376.1374389}
\BIBentrySTDinterwordspacing

\bibitem{calinescu:2011}
G.~Calinescu, C.~Chekuri, M.~P{\'{a}}l, and J.~Vondr{\'{a}}k, ``Maximizing a
  monotone submodular function subject to a matroid constraint,'' \emph{{SIAM}
  Journal on Computing}, 2011.

\bibitem{calinescu:2007}
G.~Calinescu, C.~Chekuri, M.~P{\'a}l, and J.~Vondr{\'a}k, ``Maximizing a
  submodular set function subject to a matroid constraint,'' in
  \emph{International Conference on Integer Programming and Combinatorial
  Optimization}.\hskip 1em plus 0.5em minus 0.4em\relax Springer, 2007, pp.
  182--196.

\bibitem{ageev:2004}
\BIBentryALTinterwordspacing
A.~Ageev and M.~Sviridenko, ``Pipage rounding: A new method of constructing
  algorithms with proven performance guarantee,'' \emph{Journal of
  Combinatorial Optimization}, vol.~8, no.~3, pp. 307--328, Sep. 2004.
  [Online]. Available: \url{https://doi.org/10.1023/b:joco.0000038913.96607.c2}
\BIBentrySTDinterwordspacing

\bibitem{gandhi:2006}
\BIBentryALTinterwordspacing
R.~Gandhi, S.~Khuller, S.~Parthasarathy, and A.~Srinivasan, ``Dependent
  rounding and its applications to approximation algorithms,'' \emph{J. ACM},
  vol.~53, no.~3, p. 324–360, May 2006. [Online]. Available:
  \url{https://doi.org/10.1145/1147954.1147956}
\BIBentrySTDinterwordspacing

\bibitem{topkis:1978}
\BIBentryALTinterwordspacing
D.~M. Topkis, ``Minimizing a submodular function on a lattice,''
  \emph{Operations Research}, vol.~26, no.~2, pp. 305--321, 1978. [Online].
  Available:
  \url{https://EconPapers.repec.org/RePEc:inm:oropre:v:26:y:1978:i:2:p:305-321}
\BIBentrySTDinterwordspacing

\bibitem{gottschalk:2015}
C.~Gottschalk and B.~Peis, ``Submodular function maximization on the bounded
  integer lattice,'' in \emph{International Workshop on Approximation and
  Online Algorithms}.\hskip 1em plus 0.5em minus 0.4em\relax Springer, 2015,
  pp. 133--144.

\bibitem{bach:2019}
F.~Bach, ``Submodular functions: from discrete to continuous domains,''
  \emph{Mathematical Programming}, vol. 175, no.~1, pp. 419--459, 2019.

\bibitem{sahin:2020b}
A.~Sahin, Y.~Bian, J.~Buhmann, and A.~Krause, ``From sets to multisets:
  Provable variational inference for probabilistic integer submodular models,''
  in \emph{International Conference on Machine Learning}.\hskip 1em plus 0.5em
  minus 0.4em\relax PMLR, 2020, pp. 8388--8397.

\bibitem{badanidiyuru:2014b}
A.~Badanidiyuru and J.~Vondr\'{a}k, ``Fast algorithms for maximizing submodular
  functions,'' in \emph{Proceedings of the Twenty-Fifth Annual ACM-SIAM
  Symposium on Discrete Algorithms}, ser. SODA '14.\hskip 1em plus 0.5em minus
  0.4em\relax USA: Society for Industrial and Applied Mathematics, 2014, p.
  1497–1514.

\bibitem{motwani1995randomized}
R.~Motwani and P.~Raghavan, \emph{Randomized algorithms}.\hskip 1em plus 0.5em
  minus 0.4em\relax Cambridge university press, 1995.

\end{thebibliography}

\end{document}